\pdfoutput=1
\documentclass[a4paper,USenglish,cleveref,autoref,thm-restate]{lipics-v2021}
\nolinenumbers
\hideLIPIcs

\usepackage{amsmath,amssymb}
\usepackage{xspace}
\usepackage{booktabs}
\usepackage{tikz}
\definecolor{wspdblue}{HTML}{1F77B4}
\definecolor{wspdorange}{HTML}{E66100}

\newcommand{\eps}{\varepsilon}
\DeclareMathOperator{\diam}{diam}
\DeclareMathOperator{\dist}{dist}
\newcommand{\opt}{\mathrm{OPT}}
\newcommand{\R}{\mathbb{R}}
\newcommand{\W}{\mathcal{W}}
\newcommand{\ball}{\mathbb{B}}
\newcommand{\thetacell}{\theta_{\square}}
\newcommand{\minwspd}{\textsc{minWSPD}\xspace}
\crefname{observation}{Observation}{Observations}
\Crefname{observation}{Observation}{Observations}

\title{A Bicriteria Approximation Scheme for the Minimum Well-Separated Pair Decomposition}
\titlerunning{A Bicriteria Approximation Scheme for the Minimum WSPD}

\author{Guilherme D. da Fonseca}{Aix-Marseille Universit\'e and LIS, France}{}{https://orcid.org/0000-0002-9807-028X}{}
\authorrunning{G.\,D. da Fonseca}
\Copyright{Guilherme D. da Fonseca}


\acknowledgements{Use of AI: This paper was prepared with substantial help from the AI assistant Claude (Anthropic). It took part in developing the algorithms and proofs, checking the references, and writing the text. The author checked all results and is responsible for the content.}

\EventEditors{}
\EventNoEds{0}
\EventLongTitle{}
\EventShortTitle{}
\EventAcronym{}
\EventYear{}
\EventDate{}
\EventLocation{}
\EventLogo{}
\SeriesVolume{}
\ArticleNo{}

\begin{document}
\maketitle

\begin{abstract}
A well-separated pair decomposition (WSPD) of a set $P$ of $n$ points covers all $\binom n2$ pairs of points by a family of \emph{dumbbells} $\{X,Y\}$, each covering the pairs of $X\times Y$, such that the diameters of $X$ and $Y$ are at most an $\eps$-fraction of their distance. In $\R^d$, WSPDs with $O(n/\eps^d)$ dumbbells always exist, but many inputs admit much smaller ones. However, computing a WSPD of minimum size is NP-hard in the plane. For $d\ge2$, the only known approximation that preserves the separation has a factor of $2^{O(d)}$, and Buchin et al.\ (ESA 2026) asked whether a $(1+\delta)$-approximation exists.

We give a bicriteria approximation scheme: for fixed $d$ and $\delta\in(0,1]$, we compute in $f(d,\delta)\cdot n^{6}$ time a $(1+\delta)\eps$-WSPD whose size is at most $(1+\delta)$ times the minimum size of an $\eps$-WSPD. The same holds in metric spaces of bounded doubling dimension. The algorithm combines shifting over scales with the local-growth technique of Nieberg, Hurink, and Kern, applied to a graph whose vertices are the pairs of points. In contrast, we show that in general metric spaces, relaxing the separation by a factor below $2$ keeps the problem hard: for every fixed rational $\eps\le1/4$, it is NP-hard to compute such a relaxed WSPD whose size is within a factor $n^c$ of the minimum, for every $c<2/3$.
\end{abstract}

\section{Introduction}
\label{sec:intro}

Let $P$ be a set of $n$ points in a metric space $(M,\dist)$ and let $\eps>0$. For $X\subseteq P$ let $\diam(X)$ be its diameter, and for $X,Y\subseteq P$ let $\dist(X,Y)=\min_{x\in X,y\in Y}\dist(x,y)$. A \emph{dumbbell} is a set $\{X,Y\}$ of two disjoint nonempty subsets of $P$, called its \emph{sides}. It is \emph{$\eps$-separated} if
\[
  \max\{\diam(X),\diam(Y)\}\le \eps\cdot\dist(X,Y).
\]
Throughout the paper, a \emph{pair} is an unordered pair $\{p,q\}$ of distinct points of $P$, and the dumbbell $\{X,Y\}$ \emph{covers} the pairs $\{x,y\}$ with $x\in X$ and $y\in Y$. A family $\W$ of $\eps$-separated dumbbells is an \emph{$\eps$-WSPD} of $P$ if it covers all $\binom n2$ pairs of $P$. Its \emph{size} is $|\W|$. It is a \emph{partition} WSPD if every pair is covered exactly once. Unless stated otherwise, WSPDs in this paper are \emph{cover} WSPDs, that is, pairs may be covered several times.

Since their introduction by Callahan and Kosaraju~\cite{CK95}, WSPDs have become a basic tool for proximity problems, including spanners, approximate minimum spanning trees, nearest neighbors, and distance oracles; see~\cite{NS07,Smi18}. In $\R^d$, a partition $\eps$-WSPD of size $O(n/\eps^d)$ can be computed in $O(n\log n+n/\eps^d)$ time~\cite{CK95}, and similar bounds hold in doubling metrics~\cite{HM06}. On many inputs, however, much smaller WSPDs exist. For instance, the exponentially spread points of \cref{sec:lower} admit a WSPD of size $n-1$, which is optimal by \cref{lem:mst}. Since the size of a WSPD determines the size of the resulting spanners and the cost of many downstream algorithms, it is natural to ask for a smallest one.
In \minwspd, we are given $P$ and $\eps\in(0,1)$ and we must compute an $\eps$-WSPD of minimum size. We write $\opt_\eps$, or $\opt$ when $\eps$ is clear, for this minimum size of a cover WSPD.

\subparagraph*{Previous work.}
Buchin et al.~\cite{BCHKLPRW26} initiated the study of \minwspd. They showed that the problem is NP-hard in the plane when $\eps$ depends on $n$, and that in general metric spaces, already for the fixed value $\eps=1/3$, it cannot be approximated within a factor $n^{c}$ for some constant $c>0$ unless P${}={}$NP, by padding a reduction from clique cover with many copies of the input graph at a large distance. On the real line, they gave a PTAS based on local search for pseudodisks~\cite{MR10}, an $O(n\log n)$-time $3$-approximation based on the local-ratio technique~\cite{Bar00}, and a $10$-approximation for the partition version. They also showed that the net-tree WSPD of Har-Peled and Mendel~\cite{HM06} is a $2^{O(\kappa)}$-approximation in metric spaces of doubling dimension $\kappa$, computable in $2^{O(\kappa)}(n\log n+\opt)$ expected time. Since the net-tree WSPD is a partition WSPD and $\opt_\eps$ is at most the minimum size of a partition $\eps$-WSPD, this is also a $2^{O(d)}$-approximation for the partition version in $\R^d$. They asked whether a $(1+\delta)$-approximation, or a simple $O(1)$-approximation, exists in $\R^d$ for $d\ge 2$.

Har-Peled, Raichel, and Robson~\cite[Thm.~2.14]{HRR25} construct, in polynomial time and in every finite metric space, an $\eps$-WSPD whose size is at most $7\,\opt_{\eps/33}$ for every $\eps\in(0,1)$. Their theorem is stated with unspecified constants, and the values $7$ and $33$ come from its proof. This bicriteria guarantee loses constant factors in both the size and the separation.

\begin{table}[t]
\centering\small\setlength{\tabcolsep}{3pt}
\caption{Approximation results for \minwspd. In every row, the size is compared with $\opt_\eps$, the minimum size of a cover $\eps$-WSPD. The row of~\cite{HRR25} requires $\eps<1/33$. The last row is a hardness result for every fixed rational $\eps\le1/4$, which holds for every $\tau<2$ and every $c<2/3$, unless P${}={}$NP.}
\label{tab:results}
\begin{tabular}{llllll}
\toprule
Setting & Problem & Size & Separation & Time & Reference \\
\midrule
$\R^1$ & cover & $(1+\delta)\,\opt_\eps$ & $\eps$ & $n^{O(1/\delta^2)}$ & \cite{BCHKLPRW26} \\
$\R^1$ & cover & $3\,\opt_\eps$ & $\eps$ & $O(n\log n)$ & \cite{BCHKLPRW26} \\
$\R^1$ & both & $10\,\opt_\eps$ & $\eps$ & polynomial & \cite{BCHKLPRW26} \\
Doubling dim.\ $\kappa$ & both & $2^{O(\kappa)}\opt_\eps$ & $\eps$ & $2^{O(\kappa)}(n\log n+\opt)$ exp. & \cite{BCHKLPRW26,HM06} \\
Any metric & cover & $7\,\opt_{\eps}$ & $33\eps$ & polynomial & \cite{HRR25} \\
\addlinespace[2pt]
$\R^d$, fixed $d$ & cover & $(1+\delta)\,\opt_\eps$ & $(1+\delta)\eps$ & $f(d,\delta)\,n^{6}$ & \cref{thm:ptas} \\
Doubling dim.\ $\kappa$ & cover & $(1+\delta)\,\opt_\eps$ & $(1+\delta)\eps$ & $f(\kappa,\delta)\,n^{6}$ & \cref{thm:doubling} \\
Any metric & cover & no $n^c\opt_\eps$ & $\tau\eps$, $\tau<2$ & --- & \cref{thm:hard} \\
\bottomrule
\end{tabular}
\end{table}

\subparagraph*{Our results.}
\Cref{tab:results} summarizes the known and new results.
\begin{enumerate}
\item \textbf{A bicriteria approximation scheme} (\cref{thm:ptas}). For fixed $d$ and $\delta\in(0,1]$, we compute a $(1+\delta)\eps$-WSPD of $P\subset\R^d$ of size at most $(1+\delta)\opt_\eps$ in $f(d,\delta)\cdot n^{6}$ time, where $f(d,\delta)=2^{2^{O(d(1/\delta+\log(d/\delta)))}}$. The running time does not depend on $\eps$. Equivalently, we compute an $\eps$-WSPD of size at most $(1+\delta)\opt_{\eps/(1+\delta)}$. In fixed-dimensional Euclidean space, this improves the bicriteria factors $(7,33)$ of Har-Peled et al.~\cite{HRR25} to $(1+\delta,1+\delta)$, at the cost of a much larger running time, whereas their result holds in every metric space. The question of Buchin et al., a $(1+\delta)$-approximation that keeps the separation $\eps$, remains open (see \cref{rem:ptas}). The scheme extends to metric spaces of fixed doubling dimension $\kappa$ (\cref{thm:doubling}). There, relaxing the separation by a factor $1+\delta$ brings the size factor from the $2^{O(\kappa)}$ of the net-tree WSPD down to $1+\delta$, with the same running time bound as in $\R^d$, with $\kappa$ in place of $d$.

\item \textbf{Limits in general metrics} (\cref{thm:hard}). By a direct reduction from clique cover, for every rational $\eps$ and every $\tau<2$ with $\tau\eps<1/2$ (in particular, every $\tau<2$ when $\eps\le1/4$), unless P${}={}$NP, no polynomial-time algorithm computes a $\tau\eps$-WSPD of size $\opt_\eps$ in general metric spaces, and, using the hardness of approximating the chromatic number, none computes one of size at most $n^c\opt_\eps$ for any $c<2/3$. Hence, unless P${}={}$NP, no bicriteria approximation scheme exists there. This extends the hardness result of Buchin et al.~\cite{BCHKLPRW26}, which is for $\eps=1/3$ and exact separation and whose padding with many copies we reuse, with a self-contained proof and an explicit exponent. For $\eps<1/33$, Har-Peled et al.~\cite{HRR25} achieve a constant size factor with the separation factor $33$, so for small $\eps$ the smallest separation factor that allows a constant size factor lies between $2$ and $33$.
\end{enumerate}
We also observe that in every metric space every $\eps$-WSPD has at least $n-1$ dumbbells, which is tight already on the line (\cref{lem:mst}). For partition WSPDs, this follows from the Graham--Pollak theorem~\cite{GP71}, as noted by Buchin et al.~\cite{BCHKLPRW26}. This argument does not extend to covers, since $K_n$ can be covered by $\lceil\log_2 n\rceil$ bicliques; for $\eps=1$, these bicliques are $1$-separated dumbbells of the uniform metric, so the bound needs $\eps<1$. The proof adapts Lemma~3.1 of Callahan and Kosaraju~\cite{CK93}, who show for partition WSPDs in $\R^d$ that replacing every edge of a minimum spanning tree by any pair of the same dumbbell yields a spanning tree.

\subparagraph*{Overview of techniques.}
Call two pairs \emph{compatible} if some $\eps$-separated dumbbell covers both. Then \minwspd is a clique cover problem in the compatibility graph on the $\binom n2$ pairs, restricted to the cliques that are realized by dumbbells. We solve it with the local-growth technique of Nieberg, Hurink, and Kern~\cite{NHK08,NH05}, in the form used by Pirwani and Salavatipour~\cite{PS12}. Around an uncovered pair, we grow hop-balls in the compatibility graph until increasing the radius by one multiplies the optimal cover size of the ball by at most $1+\delta/3$. No dumbbell of an optimal solution covers pairs of two different resulting cores, which gives the lower bound. Two further ingredients make this work. First, a separated dumbbell only covers pairs whose lengths differ by a factor at most $1+2\eps$, so a shifting over scales in the spirit of Hochbaum and Maass~\cite{HM85} reduces the problem to bands of bounded scale ratio, in which the optimal cover size of a hop-ball grows polynomially with its radius. Second, inside a band, rounding the sides of a dumbbell to a grid of cell size $\Theta(\delta\eps)$ times the scale relaxes the separation by a factor at most $1+\delta$ and turns every local subproblem into a finite problem of size depending only on $d$ and $\delta$, which we solve by brute force.

\subparagraph*{Further related work.}
Beyond doubling metrics, WSPDs of near-linear size exist for unit-disk graph metrics~\cite{GZ05,HRR25}, low-density graphs~\cite{GW26}, and $c$-packed graphs~\cite{DGRSW25}. Relaxations of WSPDs include semi-separated pair decompositions, which require only one side to be small~\cite{Var98,AH12}, and the approximate biclique covers of Buchin et al.~\cite{BCHKLPRW26}, which drop the condition on the diameters. Since every dumbbell is a biclique of the complete graph on $P$, \minwspd is a restricted biclique cover problem. For general bipartite graphs on $n$ vertices, minimum biclique cover is hard to approximate within $n^{1-\gamma}$ for every $\gamma>0$~\cite{CHHK14}. For weighted biclique covers of $\delta$-disk bipartite graphs, Esmailpour, Le, and Sintos~\cite[Thm.~D.3]{ELS26} relax the distance threshold by a factor $1+\varphi$, as we relax the separation, and lose a factor $O((\varphi^{-d}+\log n)\log n)$ in the cost. Similarly, Har-Peled and Lee~\cite{HL12} give a PTAS for covering points by weighted fat objects that may expand by a $\delta$-fraction of their diameter. Local search gives approximation schemes for many geometric set cover and hitting set problems with pseudodisks~\cite{MR10}, and shifting and local growth give approximation schemes for covering and packing problems in geometric intersection graphs~\cite{HM85,EJS05,NHK08,PS12,DP11}.

\subparagraph*{Organization.}
\Cref{sec:prelim} introduces notation and proves the lower bound. \Cref{sec:ptas} presents the approximation scheme in $\R^d$, and \cref{sec:doubling} its extension to doubling metrics. \Cref{sec:general} shows the hardness in general metrics. We conclude with open problems in \cref{sec:conclusion}.

\section{Preliminaries}
\label{sec:prelim}

This section fixes notation and collects the simple properties of separated dumbbells that the rest of the paper relies on. A separated dumbbell stays separated when its sides shrink, and all the pairs it covers have almost the same length, with the endpoints on each side close to each other. We then define the compatibility graph, in which \minwspd becomes a clique cover problem, and prove a lower bound on the size of every WSPD.

For a pair $e=\{x_e,y_e\}$, its \emph{scale} is $s(e)=\dist(x_e,y_e)$. We write $\ball(z,R)$ for the closed ball of radius $R$ centered at~$z$. For a set $S$ of pairs, $\theta_\eps(S)$ denotes the minimum number of $\eps$-separated dumbbells that cover all pairs of $S$, so $\opt_\eps=\theta_\eps(\binom P2)$. We use the following two basic facts, which follow directly from the definitions.

\begin{observation}[Heredity]\label{obs:hered}
If $\{X,Y\}$ is $\eps$-separated and $\emptyset\neq X'\subseteq X$, $\emptyset\ne Y'\subseteq Y$, then $\{X',Y'\}$ is $\eps$-separated.
\end{observation}

\begin{observation}[Locality]\label{obs:local}
Let $\{X,Y\}$ be an $\eps$-separated dumbbell with $\dist(X,Y)=\ell$ that covers a pair $e$. Then $\ell\le s(e)\le (1+2\eps)\ell$. Moreover, $X\subseteq \ball(x,\eps\ell)$ for every $x\in X$, and similarly for $Y$.
\end{observation}
\begin{proof}
We have $s(e)\ge\ell$ by definition, and $s(e)\le \diam(X)+\ell+\diam(Y)\le(1+2\eps)\ell$ by the triangle inequality. The second claim follows from $\diam(X)\le\eps\ell$.
\end{proof}

Two pairs $a$ and $b$ are \emph{compatible} if some $\eps$-separated dumbbell covers both. By \cref{obs:hered}, $a$ and $b$ are compatible if and only if at least one of the dumbbells $\{\{x_a,x_b\},\{y_a,y_b\}\}$ and $\{\{x_a,y_b\},\{y_a,x_b\}\}$ has disjoint sides and is $\eps$-separated. Hence, compatibility can be tested with $O(1)$ distance computations. The \emph{compatibility graph} $H$ has the pairs as vertices, and an edge between any two compatible pairs. The pairs covered by an $\eps$-separated dumbbell form a clique in $H$.

\subsection{A lower bound for cover WSPDs}
\label{sec:lower}

A minimum spanning tree of $P$ has $n-1$ edges, and we show that no dumbbell covers two of them. Intuitively, if a dumbbell covered two tree edges, then the tree paths between the endpoints on each side would be short, because the sides are small. These two paths and the shorter edge would connect the endpoints of the longer edge without using it, which is impossible in a tree (\cref{fig:mst}).

\begin{figure}[htb]
\centering
\begin{tikzpicture}[font=\small]
\draw[black!40, dashed] (0.012,0.525) circle[radius=0.782];
\node[black!70] at (0.012,-0.507) {$X$};
\draw[black!40, dashed] (5.45,0.575) circle[radius=0.893];
\node[black!70] at (5.45,-0.568) {$Y$};
\draw[wspdorange, very thick, densely dashed] (0.1,1.05) -- (5.3,1.2);
\draw[wspdblue, very thick] (0.1,1.05) -- (0.35,0.55) -- (0,0) -- (5.2,0.15) -- (5.45,0.65) -- (5.3,1.2);
\fill[black] (0,0) circle[radius=1.4pt];
\fill[black] (0.35,0.55) circle[radius=1.4pt];
\fill[black] (0.1,1.05) circle[radius=1.4pt];
\fill[black] (-0.4,0.5) circle[radius=1.4pt];
\fill[black] (5.2,0.15) circle[radius=1.4pt];
\fill[black] (5.45,0.65) circle[radius=1.4pt];
\fill[black] (5.3,1.2) circle[radius=1.4pt];
\fill[black] (5.85,0.3) circle[radius=1.4pt];
\node[below left, inner sep=1.5pt] at (0,0) {$p$};
\node[above left, inner sep=1.5pt] at (0.1,1.05) {$p'$};
\node[below right, inner sep=1.5pt] at (5.2,0.15) {$q$};
\node[above right, inner sep=1.5pt] at (5.3,1.2) {$q'$};
\end{tikzpicture}
\caption{The proof of \cref{lem:mst}. If one dumbbell $\{X,Y\}$ covered two edges $pq$ and $p'q'$ of a minimum spanning tree, with $\dist(p,q)\le\dist(p',q')$, then the tree paths from $p'$ to $p$ and from $q$ to $q'$ would use only edges shorter than $p'q'$. Together with $pq$, they would connect $p'$ and $q'$ (blue, solid) without using $p'q'$ (orange, dashed).}
\label{fig:mst}
\end{figure}
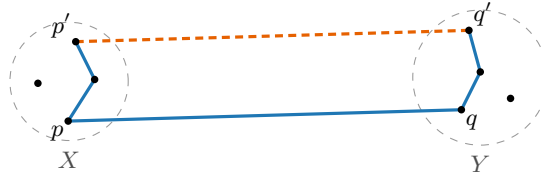

\begin{lemma}\label{lem:mst}
Let $(M,\dist)$ be a metric space, $P\subseteq M$ with $|P|=n$, and $\eps\in(0,1)$. Every $\eps$-WSPD of $P$ has at least $n-1$ dumbbells.
\end{lemma}
\begin{proof}
Let $T$ be a minimum spanning tree of the complete graph on $P$ with edge weights $\dist$. Recall the \emph{bottleneck property} of minimum spanning trees: for all $u,v\in P$, every edge on the path between $u$ and $v$ in $T$ has weight at most $\dist(u,v)$.

Let $\W$ be an $\eps$-WSPD. Map every edge $pq$ of $T$ to a dumbbell $\{X,Y\}\in\W$ with $p\in X$ and $q\in Y$. Since $T$ has $n-1$ edges, it suffices to show that this map is injective.

Suppose that two distinct edges $pq$ and $p'q'$ of $T$, with $\dist(p,q)\le \dist(p',q')$, are mapped to the same dumbbell $\{X,Y\}$. After relabeling, $p,p'\in X$ and $q,q'\in Y$. Then
\[
  \dist(p,p')\le \diam(X)\le \eps\,\dist(X,Y)\le \eps\,\dist(p,q)<\dist(p',q'),
\]
and similarly $\dist(q,q')<\dist(p',q')$. By the bottleneck property, all edges on the tree paths from $p'$ to $p$ and from $q$ to $q'$ are shorter than $p'q'$, so neither path contains $p'q'$. Together with the edge $pq\ne p'q'$, these paths form a walk from $p'$ to $q'$ in $T-p'q'$, contradicting the fact that every edge of a tree is a bridge.
\end{proof}

The bound of \cref{lem:mst} is attained already on the line: for the points $p_i=(4/\eps)^i$, $i=1,\dots,n$, the $n-1$ dumbbells $\{\{p_1,\dots,p_i\},\{p_{i+1}\}\}$, $i=1,\dots,n-1$, form an $\eps$-WSPD, as for the planar points $(i,(4/\eps)^i)$ of Buchin et al.~\cite{BCHKLPRW26}.

\section{Bicriteria Approximation Scheme}
\label{sec:ptas}

The scheme follows the local-growth approach for covering problems in geometric intersection graphs. Starting from an uncovered pair, it grows a neighborhood in the compatibility graph until one more step of growth barely increases the cost of covering it. It then covers the enlarged neighborhood optimally, removes the covered pairs, and repeats. The few extra dumbbells spent on the outer layer buy a useful property: every pair compatible with the inner neighborhood is removed, so no dumbbell covers pairs of two different inner neighborhoods, and an optimal WSPD pays separately for each of them. For this to give a factor $1+\delta$ in polynomial time, the growth must stop after a bounded number of steps, and each neighborhood must be covered optimally in bounded time. The first property requires that all pairs have comparable lengths, and the second requires a finite set of candidate dumbbells, which is where we relax the separation.

\begin{theorem}\label{thm:ptas}
There is a function $f$ such that the following holds. Let $d\ge1$ and $\delta\in(0,1]$. Given a set $P$ of $n$ points in $\R^d$ and $\eps\in(0,1)$, one can compute in $f(d,\delta)\cdot n^{6}$ time a $(1+2\delta/3)\eps$-WSPD $\W$ of $P$ with $|\W|\le(1+\delta)\opt_\eps$. Moreover, $f(d,\delta)=2^{2^{O(d(1/\delta+\log(d/\delta)))}}$.
\end{theorem}

Throughout this section, let $\zeta=\delta/3\le1/3$, $m=\lceil 1/\zeta\rceil$, $\eps'=(1+2\zeta)\eps$, and $\Lambda=(1+2\eps)^m\le 3^m$ (as $\eps<1$). The algorithm has three layers. It partitions the pairs into bands of bounded scale ratio by shifting over scales (\cref{sec:bands}). Inside each band, the dumbbells are rounded to a grid, which makes local subproblems finite and is the only step that relaxes the separation from $\eps$ to $\eps'$ (\cref{sec:rounding}). Finally, each band is covered by local growth in the compatibility graph (\cref{sec:growth}).

\subsection{Shifting over scales}
\label{sec:bands}

By \cref{obs:local}, a dumbbell covers only pairs of almost the same length. We therefore group the pairs by length into bands, each spanning a bounded ratio of lengths, and cover each band separately. A dumbbell of an optimal WSPD that covers pairs of two bands is then paid twice, but this happens only for dumbbells whose pairs straddle a band boundary. As Hochbaum and Maass~\cite{HM85} do with strips, we try $m$ shifted positions of the boundaries. For some of them, only a $1/m$ fraction of the optimal dumbbells is paid twice.

If $n\le1$, the empty family is an $\eps$-WSPD, so we assume $n\ge2$. Let $s_0=\min_e s(e)$ and define the \emph{level} of a pair $e$ as $\lambda(e)=\lfloor\log_{1+2\eps}(s(e)/s_0)\rfloor$. For an offset $o\in\{0,\dots,m-1\}$, the pairs with $\lfloor(\lambda(e)-o)/m\rfloor=j$ form the \emph{band} $E^o_j$. Each band consists of $m$ consecutive levels, so the scales of its pairs lie in $[\sigma_j,\Lambda\sigma_j)$ with $\sigma_j=s_0(1+2\eps)^{jm+o}$. We assume the real RAM model with the floor function and logarithms. This model is too powerful in general~\cite{Sch79}, but we use these operations only to compute levels and to locate grid cells (\cref{sec:rounding}), on quantities computed from the input and never to encode or extract bits. The algorithm processes only the nonempty bands, of which there are at most $\binom n2$ per offset, so its running time does not depend on the spread of $P$.

For each offset $o$ and each nonempty band $j$, the algorithm computes a family $\W^o_j$ of $\eps'$-separated dumbbells covering $E^o_j$ (\cref{sec:growth}). It outputs $\W=\bigcup_j\W^o_j$ for an offset $o$ that minimizes $\sum_j|\W^o_j|$.

\begin{lemma}\label{lem:bands}
There is an offset $o$ with $\sum_j\theta_\eps(E^o_j)\le(1+1/m)\opt_\eps$.
\end{lemma}
\begin{proof}
Fix a minimum $\eps$-WSPD $\W^*$ and an offset $o$. For every band $E^o_j$, the dumbbells of $\W^*$ that cover at least one pair of $E^o_j$ cover all of $E^o_j$. Therefore $\sum_j\theta_\eps(E^o_j)\le\sum_{Z\in\W^*}\beta^o(Z)$, where $\beta^o(Z)$ is the number of bands (for offset $o$) that contain a pair covered by $Z$.

By \cref{obs:local}, the pairs covered by $Z$ have scales in $[\ell,(1+2\eps)\ell]$, so their levels take at most two consecutive values $\lambda_0$ and $\lambda_0+1$. Hence $\beta^o(Z)\le 2$, and $\beta^o(Z)=2$ only for the single offset $o\equiv\lambda_0+1\pmod m$ (\cref{fig:bands}). Averaging over the $m$ offsets gives $\frac1m\sum_o\sum_Z\beta^o(Z)\le (1+1/m)|\W^*|$, so some offset $o$ satisfies $\sum_Z\beta^o(Z)\le(1+1/m)\opt_\eps$.
\end{proof}

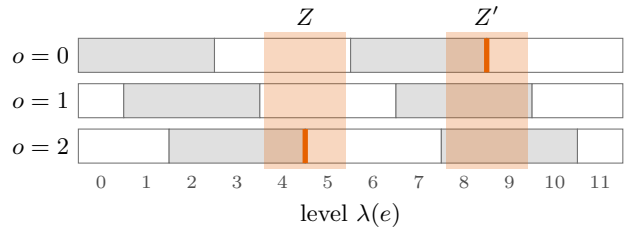
\begin{figure}[htb]
\centering
\begin{tikzpicture}[font=\small]
\filldraw[fill=black!12, draw=black!55] (0,0) rectangle (1.8,0.45);
\filldraw[fill=white, draw=black!55] (1.8,0) rectangle (3.6,0.45);
\filldraw[fill=black!12, draw=black!55] (3.6,0) rectangle (5.4,0.45);
\filldraw[fill=white, draw=black!55] (5.4,0) rectangle (7.2,0.45);
\node[left] at (0,0.225) {$o=0$};
\draw[wspdorange, line width=2pt] (5.4,0) -- (5.4,0.45);
\filldraw[fill=white, draw=black!55] (0,-0.6) rectangle (0.6,-0.15);
\filldraw[fill=black!12, draw=black!55] (0.6,-0.6) rectangle (2.4,-0.15);
\filldraw[fill=white, draw=black!55] (2.4,-0.6) rectangle (4.2,-0.15);
\filldraw[fill=black!12, draw=black!55] (4.2,-0.6) rectangle (6,-0.15);
\filldraw[fill=white, draw=black!55] (6,-0.6) rectangle (7.2,-0.15);
\node[left] at (0,-0.375) {$o=1$};
\filldraw[fill=white, draw=black!55] (0,-1.2) rectangle (1.2,-0.75);
\filldraw[fill=black!12, draw=black!55] (1.2,-1.2) rectangle (3,-0.75);
\filldraw[fill=white, draw=black!55] (3,-1.2) rectangle (4.8,-0.75);
\filldraw[fill=black!12, draw=black!55] (4.8,-1.2) rectangle (6.6,-0.75);
\filldraw[fill=white, draw=black!55] (6.6,-1.2) rectangle (7.2,-0.75);
\node[left] at (0,-0.975) {$o=2$};
\draw[wspdorange, line width=2pt] (3,-1.2) -- (3,-0.75);
\fill[wspdorange, opacity=0.25] (2.46,-1.28) rectangle (3.54,0.53);
\node[above] at (3,0.5) {$Z$};
\fill[wspdorange, opacity=0.25] (4.86,-1.28) rectangle (5.94,0.53);
\node[above] at (5.4,0.5) {$Z'$};
\node[below, black!70] at (0.3,-1.25) {\scriptsize $0$};
\node[below, black!70] at (0.9,-1.25) {\scriptsize $1$};
\node[below, black!70] at (1.5,-1.25) {\scriptsize $2$};
\node[below, black!70] at (2.1,-1.25) {\scriptsize $3$};
\node[below, black!70] at (2.7,-1.25) {\scriptsize $4$};
\node[below, black!70] at (3.3,-1.25) {\scriptsize $5$};
\node[below, black!70] at (3.9,-1.25) {\scriptsize $6$};
\node[below, black!70] at (4.5,-1.25) {\scriptsize $7$};
\node[below, black!70] at (5.1,-1.25) {\scriptsize $8$};
\node[below, black!70] at (5.7,-1.25) {\scriptsize $9$};
\node[below, black!70] at (6.3,-1.25) {\scriptsize $10$};
\node[below, black!70] at (6.9,-1.25) {\scriptsize $11$};
\node[below] at (3.6,-1.6) {level $\lambda(e)$};
\end{tikzpicture}
\caption{Bands of $m=3$ consecutive levels for the offsets $o=0,1,2$. The pairs covered by a dumbbell have at most two consecutive levels, so the dumbbell meets two bands for at most one offset (thick boundary): $o=2$ for $Z$ and $o=0$ for $Z'$.}
\label{fig:bands}
\end{figure}

\subsection{Rounding}
\label{sec:rounding}

The neighborhoods grown in \cref{sec:growth} are geometrically small, but they may contain arbitrarily many points, so the number of dumbbells that can cover them is unbounded. To solve them by brute force, we round every dumbbell to a grid whose cells are much smaller than $\eps$ times the scale of the band: each side is replaced by the union of the cells it meets. Points in the same cell then become interchangeable, and a neighborhood is described by the bounded number of cells it meets. Rounding enlarges the sides by a tiny amount, so it increases their diameters and decreases their distance only slightly. This relaxes the separation from $\eps$ to $\eps'$, but it never increases the number of dumbbells.

Fix a band with scales in $[\sigma,\Lambda\sigma)$ and consider the axis-parallel grid of half-open cubes $\prod_{i=1}^d[k_i\eta,(k_i+1)\eta)$, $k_i\in\mathbb{Z}$, of side
\[
  \eta=\frac{\zeta\,\eps\,\sigma}{12\sqrt d}.
\]
Every point of $\R^d$ lies in exactly one cell. A \emph{cell dumbbell} is a set $\{A,B\}$ of two nonempty finite sets of grid cells such that the regions $\bigcup A$ and $\bigcup B$ satisfy
\[
  \max\{\diam(\textstyle\bigcup A),\diam(\bigcup B)\}\le \eps'\,\dist(\bigcup A,\bigcup B).
\]
Here the diameter of a region and the distance between two regions are those of their closures. It covers a pair $e$ if one endpoint of $e$ lies in $\bigcup A$ and the other in $\bigcup B$. If both sides contain points of $P$, the dumbbell $\{P\cap\bigcup A,\,P\cap\bigcup B\}$ is $\eps'$-separated, because taking subsets does not increase diameters and does not decrease distances. We identify a cell dumbbell with this dumbbell (\cref{fig:round}). For a set $S$ of pairs of the band, $\thetacell(S)$ denotes the minimum number of cell dumbbells covering $S$.

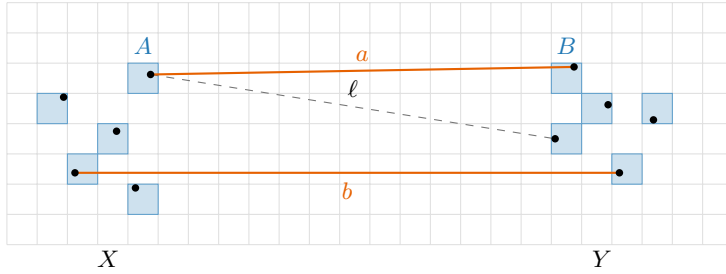
\begin{figure}[htb]
\centering
\begin{tikzpicture}[font=\small]
\draw[black!12] (-1.2,-1.2) grid[step=0.4] (8.4,2);
\fill[wspdblue!22] (-0.8,0.4) rectangle (-0.4,0.8);
\draw[wspdblue!70] (-0.8,0.4) rectangle (-0.4,0.8);
\fill[wspdblue!22] (-0.4,-0.4) rectangle (0,0);
\draw[wspdblue!70] (-0.4,-0.4) rectangle (0,0);
\fill[wspdblue!22] (0,0) rectangle (0.4,0.4);
\draw[wspdblue!70] (0,0) rectangle (0.4,0.4);
\fill[wspdblue!22] (0.4,-0.8) rectangle (0.8,-0.4);
\draw[wspdblue!70] (0.4,-0.8) rectangle (0.8,-0.4);
\fill[wspdblue!22] (0.4,0.8) rectangle (0.8,1.2);
\draw[wspdblue!70] (0.4,0.8) rectangle (0.8,1.2);
\node[wspdblue, above] at (0.6,1.2) {$A$};
\fill[wspdblue!22] (6,0) rectangle (6.4,0.4);
\draw[wspdblue!70] (6,0) rectangle (6.4,0.4);
\fill[wspdblue!22] (6,0.8) rectangle (6.4,1.2);
\draw[wspdblue!70] (6,0.8) rectangle (6.4,1.2);
\fill[wspdblue!22] (6.4,0.4) rectangle (6.8,0.8);
\draw[wspdblue!70] (6.4,0.4) rectangle (6.8,0.8);
\fill[wspdblue!22] (6.8,-0.4) rectangle (7.2,0);
\draw[wspdblue!70] (6.8,-0.4) rectangle (7.2,0);
\fill[wspdblue!22] (7.2,0.4) rectangle (7.6,0.8);
\draw[wspdblue!70] (7.2,0.4) rectangle (7.6,0.8);
\node[wspdblue, above] at (6.2,1.2) {$B$};
\draw[black!55, dashed] (0.7,1.05) -- node[pos=0.5, above, black] {$\ell$} (6.05,0.2);
\draw[wspdorange, thick] (0.7,1.05) -- node[pos=0.5, above] {$a$} (6.3,1.15);
\draw[wspdorange, thick] (-0.3,-0.25) -- node[pos=0.5, below] {$b$} (6.9,-0.25);
\fill[black] (0.25,0.3) circle[radius=1.4pt];
\fill[black] (-0.45,0.75) circle[radius=1.4pt];
\fill[black] (0.7,1.05) circle[radius=1.4pt];
\fill[black] (-0.3,-0.25) circle[radius=1.4pt];
\fill[black] (0.5,-0.45) circle[radius=1.4pt];
\fill[black] (6.05,0.2) circle[radius=1.4pt];
\fill[black] (6.75,0.65) circle[radius=1.4pt];
\fill[black] (6.3,1.15) circle[radius=1.4pt];
\fill[black] (6.9,-0.25) circle[radius=1.4pt];
\fill[black] (7.35,0.45) circle[radius=1.4pt];
\node[below] at (0.14,-1.15) {$X$};
\node[below] at (6.67,-1.15) {$Y$};
\end{tikzpicture}
\caption{A dumbbell $\{X,Y\}$ covering the pairs $a$ and $b$, and its cell dumbbell $\{A,B\}$.}
\label{fig:round}
\end{figure}

\begin{lemma}\label{lem:round}
Let $\{X,Y\}$ be an $\eps$-separated dumbbell covering a pair $e$ with $s(e)\ge\sigma$, and let $A$ and $B$ be the sets of cells meeting $X$ and $Y$. Then $\{A,B\}$ is a cell dumbbell that covers every pair covered by $\{X,Y\}$. In particular, $\thetacell(S)\le\theta_\eps(S)$ for every set $S$ of pairs of the band.
\end{lemma}
\begin{proof}
Let $\ell=\dist(X,Y)$. By \cref{obs:local}, $\ell\ge\sigma/(1+2\eps)\ge\sigma/3$, so $\eta\sqrt d\le\zeta\eps\ell/4$. Every point of $\bigcup A$ is within distance $\eta\sqrt d$ of a point of $X$. Therefore
\[
  \diam(\textstyle\bigcup A)\le\eps\ell+2\eta\sqrt d\le\eps\ell(1+\zeta/2)
  \quad\text{and}\quad
  \dist(\textstyle\bigcup A,\bigcup B)\ge\ell-2\eta\sqrt d\ge\ell(1-\zeta/2),
\]
and the same holds for $B$. Since $1+\zeta/2\le(1+2\zeta)(1-\zeta/2)$ for $\zeta\le1/2$, the dumbbell $\{A,B\}$ is $\eps'$-separated. Finally, $X\subseteq\bigcup A$ and $Y\subseteq\bigcup B$, so $\{A,B\}$ covers every pair covered by $\{X,Y\}$.
\end{proof}

\begin{lemma}\label{lem:grid}
Let $\rho\ge1$, let $u,w\in\R^d$, and let $S$ be a set of pairs of the band whose endpoints lie in $\ball(u,R)\cup\ball(w,R)$, where $R=\rho\,\eps\Lambda\sigma$. Let $Q$ be the set of cells that contain an endpoint of a pair of $S$. Then $|Q|\le K=2(24\sqrt d\,\rho\Lambda/\zeta+2)^d$, and:
\begin{enumerate}
\item $\thetacell(S)\le K^2$;
\item $\thetacell(S)$ and an optimal family of cell dumbbells can be computed in $g(d,\zeta,\rho)+O(|S|\log|S|)$ time, for some function $g$ that does not depend on $n$ or $\eps$.
\end{enumerate}
\end{lemma}
\begin{proof}
A ball of radius $R$ meets at most $(2R/\eta+2)^d$ cells, which gives the bound on $|Q|$.

(1) For a pair $e\in S$ with endpoints in cells $c$ and $c'$, the set $\{\{c\},\{c'\}\}$ covers $e$. It is a cell dumbbell because $\diam(c)=\eta\sqrt d\le\zeta\eps\sigma/12$ and $\dist(c,c')\ge\sigma-2\eta\sqrt d\ge\sigma/2$ (in particular, $c\ne c'$).

(2) A cell dumbbell that covers no pair of $S$ can be dropped. If a cell dumbbell $\{A,B\}$ covers a pair of $S$, then $A\cap Q$ and $B\cap Q$ are nonempty, and $\{A\cap Q,B\cap Q\}$ is a cell dumbbell that covers the same pairs of $S$, because shrinking the regions does not increase diameters and does not decrease distances. So we may assume $A,B\subseteq Q$. Every cell dumbbell covering a pair of $S$ satisfies $\dist(\bigcup A,\bigcup B)\le\Lambda\sigma$, hence $\diam(\bigcup A)\le\eps'\Lambda\sigma\le2\eps\Lambda\sigma$. Therefore, if $c\in A$, then every cell of $A$ lies in the cube of $(2D/\eta+3)^d$ cells centered at $c$, where $D=2\eps\Lambda\sigma$, and so $A$ is contained in a set of at most $K_0=(48\sqrt d\,\Lambda/\zeta+3)^d$ cells determined by any one of its cells. There are at most $K\cdot2^{K_0}$ choices for each side. The diameter of a union of cells and the distance between two unions of cells are attained by pairs of cells and can be computed from their corners, so testing a candidate takes time polynomial in $K_0$. Whether a cell dumbbell covers a pair of $S$ depends only on the cells that contain the endpoints of the pair, so by sorting in $O(|S|\log|S|)$ time we replace $S$ by its set of at most $K^2$ distinct pairs of cells. By~(1), an optimal solution uses at most $K^2$ cell dumbbells, and exhaustive search over families of at most $K^2$ candidates takes time bounded by a function of $d$, $\zeta$, and $\rho$ only. Finally, the scale $\eps\sigma$ cancels in $R/\eta$ and $D/\eta$, and $\Lambda\le3^m$, so after replacing $\Lambda$ by $3^m$ none of these bounds depend on $\eps$.
\end{proof}

\subsection{Local growth inside a band}
\label{sec:growth}

We now cover a single band. The algorithm grows a ball around an uncovered pair in the compatibility graph, one hop at a time, as long as each hop multiplies the optimal cover size of the ball by more than $1+\zeta$. The endpoints of the pairs in a ball of radius $r$ lie within distance $r\eps\Lambda\sigma$ of the endpoints of the starting pair, where $\sigma$ is the scale of the band, so after rounding, the optimal cover size grows only polynomially with $r$, and the exponential growth must stop after a bounded number of hops. The last ball, the core, is then covered together with one more hop, which costs at most a factor $1+\zeta$ more than covering the core alone, and the pairs covered are removed. In particular, every pair compatible with the core is removed, so no dumbbell covers pairs of two different cores, and an optimal cover pays for each core separately.

Fix a band $E=E^o_j$ with scales in $[\sigma,\Lambda\sigma)$. The algorithm maintains the set $U\subseteq E$ of uncovered pairs, initially $U=E$. Let $H[U]$ be the compatibility graph induced by $U$, and let $N^r(v)$ be the set of pairs at hop distance at most $r$ from $v$ in $H[U]$. While $U\ne\emptyset$, the algorithm performs a \emph{round} consisting of the following steps:
\begin{enumerate}
\item pick any $v\in U$;
\item let $r\ge0$ be the smallest integer with $\thetacell(N^{r+1}(v))\le(1+\zeta)\,\thetacell(N^r(v))$;
\item add to $\W^o_j$ a minimum family of cell dumbbells covering $N^{r+1}(v)$, and remove from $U$ all pairs that it covers.
\end{enumerate}
The integer $r$ in step~2 exists by \cref{lem:radius} below. We call $N^r(v)$ the \emph{core} of the round. In round $i$, we write $U_i$ for the set $U$ at the start of the round, and $v_i$, $r_i$, and $C_i=N^{r_i}(v_i)$ for the objects chosen in the round. All pairs of $N^{r_i+1}(v_i)$ are removed from $U$ in round $i$.

\begin{lemma}\label{lem:cores}
Let $C_1,C_2,\dots$ be the cores in the order of the rounds. Then $\sum_i\thetacell(C_i)\le\theta_\eps(E)$, and consequently $|\W^o_j|\le(1+\zeta)\,\theta_\eps(E)$.
\end{lemma}
\begin{proof}
Let $\W^*$ be a minimum family of $\eps$-separated dumbbells covering $E$. We claim that no dumbbell $Z\in\W^*$ covers pairs $a\in C_i$ and $b\in C_k$ with $i<k$. Suppose it does (\cref{fig:cores}). Then $a$ and $b$ are compatible. Since $b\in U_k\subseteq U_i$ and $a\in C_i\subseteq U_i$, the pairs $a$ and $b$ are adjacent in $H[U_i]$. Hence $b\in N^{r_i+1}(v_i)$, so $b$ was removed from $U$ in round $i<k$, a contradiction.

Therefore, the sets $\W^*_i=\{Z\in\W^*: Z\text{ covers a pair of }C_i\}$ are pairwise disjoint. Since $\W^*_i$ covers $C_i$, \cref{lem:round} gives $|\W^*_i|\ge\theta_\eps(C_i)\ge\thetacell(C_i)$. Summing over $i$ gives the first claim. For the second claim, round $i$ adds $\thetacell(N^{r_i+1}(v_i))\le(1+\zeta)\,\thetacell(C_i)$ dumbbells to $\W^o_j$, and summing over the rounds gives the bound.
\end{proof}

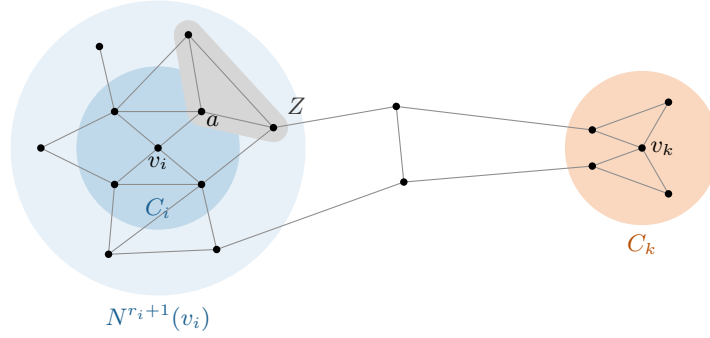
\begin{figure}[htb]
\centering
\begin{tikzpicture}[font=\small]
\fill[wspdblue!10] (0,0) circle[radius=1.95];
\fill[wspdblue!28] (0,0) circle[radius=1.08];
\fill[wspdorange!25] (6.4,0) circle[radius=1.02];
\draw[black!16, line width=11pt, line cap=round, line join=round] (0.575,0.482) -- (1.526,0.269) -- (0.401,1.497) -- cycle;
\fill[black!16] (0.575,0.482) -- (1.526,0.269) -- (0.401,1.497) -- cycle;
\node[black!80, above right, inner sep=5pt] at (1.526,0.269) {$Z$};
\draw[black!45] (0,0) -- (0.575,0.482);
\draw[black!45] (0,0) -- (-0.575,0.482);
\draw[black!45] (0,0) -- (-0.575,-0.482);
\draw[black!45] (0,0) -- (0.575,-0.482);
\draw[black!45] (0.575,0.482) -- (-0.575,0.482);
\draw[black!45] (-0.575,-0.482) -- (0.575,-0.482);
\draw[black!45] (0.575,0.482) -- (1.526,0.269);
\draw[black!45] (0.575,0.482) -- (0.401,1.497);
\draw[black!45] (-0.575,0.482) -- (0.401,1.497);
\draw[black!45] (-0.575,0.482) -- (-0.775,1.342);
\draw[black!45] (-0.575,0.482) -- (-1.55,0);
\draw[black!45] (-0.575,-0.482) -- (-1.55,0);
\draw[black!45] (-0.575,-0.482) -- (-0.655,-1.405);
\draw[black!45] (0.575,-0.482) -- (-0.655,-1.405);
\draw[black!45] (0.575,-0.482) -- (0.775,-1.342);
\draw[black!45] (0.575,-0.482) -- (1.526,0.269);
\draw[black!45] (1.526,0.269) -- (0.401,1.497);
\draw[black!45] (-0.655,-1.405) -- (0.775,-1.342);
\draw[black!45] (1.526,0.269) -- (3.15,0.55);
\draw[black!45] (0.775,-1.342) -- (3.25,-0.45);
\draw[black!45] (3.15,0.55) -- (3.25,-0.45);
\draw[black!45] (3.15,0.55) -- (5.742,0.239);
\draw[black!45] (3.25,-0.45) -- (5.742,-0.239);
\draw[black!45] (6.4,0) -- (6.75,0.606);
\draw[black!45] (6.4,0) -- (5.742,0.239);
\draw[black!45] (6.4,0) -- (5.742,-0.239);
\draw[black!45] (6.4,0) -- (6.75,-0.606);
\draw[black!45] (6.75,0.606) -- (5.742,0.239);
\draw[black!45] (5.742,-0.239) -- (6.75,-0.606);
\fill[black] (0,0) circle[radius=1.4pt];
\fill[black] (6.4,0) circle[radius=1.4pt];
\fill[black] (0.575,0.482) circle[radius=1.4pt];
\fill[black] (-0.575,0.482) circle[radius=1.4pt];
\fill[black] (-0.575,-0.482) circle[radius=1.4pt];
\fill[black] (0.575,-0.482) circle[radius=1.4pt];
\fill[black] (1.526,0.269) circle[radius=1.4pt];
\fill[black] (0.401,1.497) circle[radius=1.4pt];
\fill[black] (-0.775,1.342) circle[radius=1.4pt];
\fill[black] (-1.55,0) circle[radius=1.4pt];
\fill[black] (-0.655,-1.405) circle[radius=1.4pt];
\fill[black] (0.775,-1.342) circle[radius=1.4pt];
\fill[black] (3.15,0.55) circle[radius=1.4pt];
\fill[black] (3.25,-0.45) circle[radius=1.4pt];
\fill[black] (6.75,0.606) circle[radius=1.4pt];
\fill[black] (5.742,0.239) circle[radius=1.4pt];
\fill[black] (5.742,-0.239) circle[radius=1.4pt];
\fill[black] (6.75,-0.606) circle[radius=1.4pt];
\node[below, inner sep=3pt] at (0,0) {$v_i$};
\node[right, inner sep=3pt] at (6.4,0) {$v_k$};
\node[below right, inner sep=1.5pt] at (0.575,0.482) {$a$};
\node[wspdblue!80!black] at (0,-0.8) {$C_i$};
\node[wspdblue!80!black, below] at (0,-1.95) {$N^{r_i+1}(v_i)$};
\node[wspdorange!80!black, below] at (6.4,-1.02) {$C_k$};
\end{tikzpicture}
\caption{Round $i$ in $H[U_i]$ (schematic, $r_i=1$). The pairs of the core $C_i$ (dark) and of $N^{r_i+1}(v_i)$ (light) are removed. The pairs covered by a dumbbell $Z$ form a clique of $H$, so if $Z$ covers a pair $a\in C_i$, then all its pairs in $U_i$ lie in $N^{r_i+1}(v_i)$. Hence no dumbbell covers pairs of $C_i$ and of a later core $C_k$.}
\label{fig:cores}
\end{figure}

\begin{figure}[htb]
\centering
\begin{tikzpicture}[font=\small]
\draw[black!40, dashed] (0,0) circle[radius=0.9];
\draw[black!40, dashed] (0,0) circle[radius=1.8];
\draw[black!40, dashed] (7,0.8) circle[radius=0.9];
\draw[black!40, dashed] (7,0.8) circle[radius=1.8];
\draw[black!55, ->, >=stealth] (0,0) -- (-0.779,0.45);
\node[black!70, left] at (-0.779,0.45) {$\eps\Lambda\sigma$};
\draw[black!55, ->, >=stealth] (0,0) -- (-1.559,-0.9);
\node[black!70, left] at (-1.559,-0.9) {$2\eps\Lambda\sigma$};
\draw[wspdorange, thick, densely dotted] (1.05,1.05) -- (5.95,1.9);
\fill[wspdorange] (1.05,1.05) circle[radius=1.4pt];
\fill[wspdorange] (5.95,1.9) circle[radius=1.4pt];
\draw[wspdorange, thick, densely dotted] (-0.7,-1.45) -- (8.3,-0.3);
\fill[wspdorange] (-0.7,-1.45) circle[radius=1.4pt];
\fill[wspdorange] (8.3,-0.3) circle[radius=1.4pt];
\draw[wspdorange, thick, densely dotted] (0.95,-0.95) -- (6.6,-0.6);
\fill[wspdorange] (0.95,-0.95) circle[radius=1.4pt];
\fill[wspdorange] (6.6,-0.6) circle[radius=1.4pt];
\draw[wspdblue, thick] (0.45,0.4) -- (6.55,1.35);
\fill[wspdblue] (0.45,0.4) circle[radius=1.4pt];
\fill[wspdblue] (6.55,1.35) circle[radius=1.4pt];
\draw[wspdblue, thick] (-0.35,-0.55) -- (7.55,0.35);
\fill[wspdblue] (-0.35,-0.55) circle[radius=1.4pt];
\fill[wspdblue] (7.55,0.35) circle[radius=1.4pt];
\draw[very thick] (0,0) -- node[pos=0.5, below, sloped] {$v$} (7,0.8);
\fill[black] (0,0) circle[radius=1.4pt];
\fill[black] (7,0.8) circle[radius=1.4pt];
\node[below right] at (0,0) {$u$};
\node[below left] at (7,0.8) {$w$};
\end{tikzpicture}
\caption{Pairs at hop distance $1$ (blue, solid) and $2$ (orange, dotted) from the pair $v=\{u,w\}$ in $H$ (schematic). By the proof of \cref{lem:radius}, their endpoints lie in the dashed balls of radius $\eps\Lambda\sigma$ and $2\eps\Lambda\sigma$ around $u$ and $w$.}
\label{fig:hop}
\end{figure}
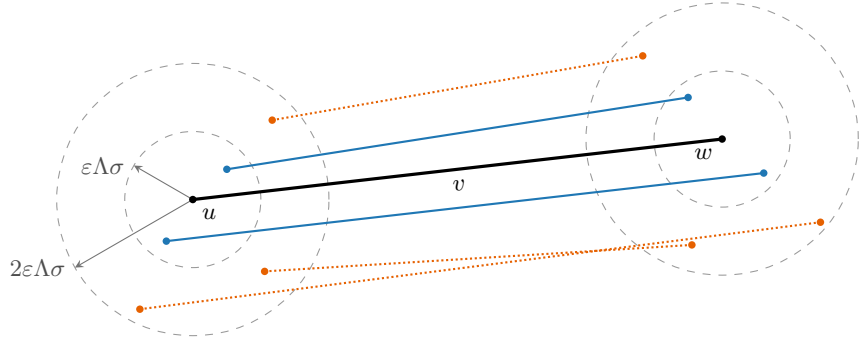

\begin{lemma}\label{lem:radius}
In every round, an integer $r$ as in step~2 exists and satisfies $r\le r_{\max}$ for some $r_{\max}=O\big((d/\zeta)(1/\zeta+\log(d/\zeta))\big)$.
\end{lemma}
\begin{proof}
Let $a,b\in E$ be compatible, and let $Z$ be an $\eps$-separated dumbbell covering both. Each side of $Z$ contains one endpoint of $a$ and one endpoint of $b$. By \cref{obs:local}, these two endpoints are at distance at most $\eps\,s(a)<\eps\Lambda\sigma$. By induction, all endpoints of pairs of $N^r(v)$ lie within distance $r\eps\Lambda\sigma$ of an endpoint of $v$ (\cref{fig:hop}).

By \cref{lem:grid}(1) with $\rho=r+1$, and since $2\le2\sqrt d\,(r+1)\Lambda/\zeta$, this implies $\thetacell(N^r(v))\le \mu_r:=4\big(26\sqrt d\,(r+1)\Lambda/\zeta\big)^{2d}$. Now let $r\ge1$ be such that the condition of step~2 fails for every $r'<r$, that is, $\thetacell(N^{r'+1}(v))>(1+\zeta)\,\thetacell(N^{r'}(v))$. Since $\thetacell(N^0(v))=\thetacell(\{v\})=1$, we get $(1+\zeta)^r<\thetacell(N^r(v))\le \mu_r$. Taking logarithms and using $\ln(1+\zeta)\ge\zeta/2$, $\ln\Lambda\le m\ln3$, and $m\le1/\zeta+1$, we get $r\le \xi\,(L+\ln(r+1))$ with $\xi=O(d/\zeta)$, $\xi\ge1$, and $L=1/\zeta+\ln(d/\zeta)\ge1$. This implies $r=O(\xi(L+\ln\xi))$, and $\ln\xi=O(L)$, so $r\le r_{\max}$ for some $r_{\max}=O\big((d/\zeta)(1/\zeta+\log(d/\zeta))\big)$. Hence the condition of step~2 holds for some $r\le r_{\max}$.
\end{proof}

\begin{proof}[Proof of \cref{thm:ptas}]
All dumbbells in the output are cell dumbbells whose cells contain points of $P$ (by the proof of \cref{lem:grid}(2)), hence $\eps'$-separated, and $\eps'=(1+2\delta/3)\eps$. Let $o^*$ be the offset of \cref{lem:bands}. Since the algorithm chooses an offset that minimizes $\sum_j|\W^o_j|$, \cref{lem:cores,lem:bands} give
\[
  |\W|\le\sum_j|\W^{o^*}_j|\le(1+\zeta)\sum_j\theta_\eps(E^{o^*}_j)\le(1+\zeta)^2\opt_\eps\le(1+3\zeta)\opt_\eps=(1+\delta)\opt_\eps,
\]
using $1/m\le\zeta$ and $\zeta\le1$.

By \cref{lem:radius}, every value $\thetacell(N^r(v))$ evaluated by the algorithm has $r\le r_{\max}+1$, and $N^r(v)$ satisfies the hypothesis of \cref{lem:grid} with $\rho=r_{\max}+2$ and $u,w$ the endpoints of $v$. Hence each such value, and the corresponding optimal family, is computed in $g'(d,\delta)+O(n^2\log n)$ time. Here, since $\log\Lambda\le m\log3$ and $\log r_{\max}=O(\log(d/\zeta))$, the numbers $K$ and $K_0$ of the proof of \cref{lem:grid} are $2^{O(d(1/\zeta+\log(d/\zeta)))}$. The exhaustive search examines at most $(K\cdot2^{K_0})^{2K^2}=2^{O(K_0K^2)}$ families, so $g'(d,\delta)=2^{2^{O(d(1/\delta+\log(d/\delta)))}}$. Building $H$ takes $O(n^4)$ compatibility tests. For each of the $m$ offsets, there are $O(n^2)$ rounds in total over all bands, and each round computes its hop-balls by one breadth-first search in $O(n^4)$ time and evaluates $O(r_{\max})$ values of $\thetacell$. The total running time is therefore $O(m\,n^6)+O(m\,r_{\max}\,n^2(g'(d,\delta)+n^2\log n))$, which is $f(d,\delta)\cdot n^{6}$, since $\zeta=\delta/3$.
\end{proof}

\begin{remark}\label{rem:ptas}
The only place where the separation is relaxed is the rounding. Suppose that, for a set $S$ of pairs as in \cref{lem:grid}, one could decide in $|S|^{h(k)}$ time whether $S$ can be covered by $k$ $\eps$-separated dumbbells, and find such a cover. Then \cref{lem:bands,lem:cores,lem:radius} hold with $\theta_\eps$ in place of $\thetacell$ throughout. The only additional step is a bound on $\theta_\eps(S)$ for such sets $S$, which replaces \cref{lem:grid}(1). For this, consider the grid of side $\eta'=\eps\sigma/(4\sqrt d)$. If a pair of the band has endpoints in cells $c$ and $c'$, then $\diam(c)=\eps\sigma/4$ and $\dist(c,c')\ge\sigma/2$, so $\{P\cap c,P\cap c'\}$ is an $\eps$-separated dumbbell that covers the pair. As in \cref{lem:grid}, this gives $\theta_\eps(S)\le K'^2$ with $K'=2(8\sqrt d\,\rho\Lambda+2)^d$. Hence the local-growth radius is bounded by a function of $d$ and $\delta$, every local instance has optimum at most $K'^2$ for the corresponding $\rho$, and we would obtain an $\eps$-WSPD of size at most $(1+\delta)\opt_\eps$ in $n^{O(h(K'^2))}$ time, that is, a PTAS that does not relax $\eps$. For clique partition in unit disk graphs, such local solvers rely on the fact that optimal clusterings in the plane can be separated by lines~\cite{CRW91}; see~\cite{PS12,DP11}. We do not know an analogue of this property for WSPDs, even in the plane.
\end{remark}

\section{Doubling Metrics}
\label{sec:doubling}

The scheme of \cref{sec:ptas} uses the geometry of $\R^d$ only in the rounding step. The grid needs coordinates, but the analysis uses only two of its properties: every cell has small diameter, and a ball meets few cells. Both properties also hold for a net of the points in a metric space of bounded doubling dimension. We therefore replace the grid cells by clusters of points around the net points, and the rest of the algorithm and of its analysis carries over.

A metric space has \emph{doubling dimension at most} $\kappa$ if every ball of radius $2R$ can be covered by $2^\kappa$ balls of radius $R$, and its doubling dimension is the smallest such $\kappa$. We use the doubling dimension of $P$ itself, viewed as a metric space, so all balls below are centered at points of $P$. The doubling dimension of $\R^d$ is $\Theta(d)$, and a subset has at most twice the doubling dimension of the space containing it. We assume that distances can be evaluated in $O(1)$ time.

\begin{theorem}\label{thm:doubling}
There is a function $f$ such that the following holds. Let $\kappa\ge1$ and $\delta\in(0,1]$. Given a set $P$ of $n$ points of doubling dimension at most $\kappa$ in a metric space and $\eps\in(0,1)$, one can compute in $f(\kappa,\delta)\cdot n^{6}$ time a $(1+2\delta/3)\eps$-WSPD $\W$ of $P$ with $|\W|\le(1+\delta)\opt_\eps$. Moreover, $f(\kappa,\delta)=2^{2^{O(\kappa(1/\delta+\log(\kappa/\delta)))}}$, and the algorithm does not need to know $\kappa$.
\end{theorem}

The algorithm is that of \cref{sec:ptas}, with the same $\zeta=\delta/3$, $m$, $\eps'$, and $\Lambda$, and with the grid replaced by a net. \Cref{obs:hered,obs:local,lem:bands} and the first paragraph of the proof of \cref{lem:radius} use only the triangle inequality, so they hold verbatim. \Cref{lem:cores} also holds once $\thetacell$ is redefined as below, with \cref{lem:round-net} in place of \cref{lem:round}. The only geometric part of \cref{sec:ptas} is \cref{sec:rounding}, and we now replace \cref{lem:round,lem:grid}.

\subparagraph*{Nets and cells.}
Fix a band with scales in $[\sigma,\Lambda\sigma)$ and let $\eta=\zeta\eps\sigma/48$. Let $P_\eta\subseteq P$ be a maximal subset with pairwise distances greater than $\eta$, computed greedily in $O(n^2)$ time. By maximality, every point of $P$ is within distance $\eta$ of $P_\eta$. Assign every point of $P$ to a nearest point of $P_\eta$. The \emph{cell} $\Gamma(q)$ of $q\in P_\eta$ is the set of points assigned to $q$, so $\Gamma(q)\subseteq\ball(q,\eta)$ and $\diam(\Gamma(q))\le2\eta$. For $A\subseteq P_\eta$, let $\Gamma(A)=\bigcup_{q\in A}\Gamma(q)$. A \emph{cell dumbbell} is now a set $\{A,B\}$ of two nonempty subsets of $P_\eta$ such that $\{\Gamma(A),\Gamma(B)\}$ is an $\eps'$-separated dumbbell, which we identify with $\{A,B\}$. As before, $\thetacell(S)$ denotes the minimum number of cell dumbbells covering a set $S$ of pairs. Unlike grid cells, these cells are sets of points of $P$, so separation is tested directly on the points.

\begin{lemma}\label{lem:round-net}
Let $\{X,Y\}$ be an $\eps$-separated dumbbell covering a pair $e$ with $s(e)\ge\sigma$, and let $A,B\subseteq P_\eta$ be the net points whose cells meet $X$ and $Y$. Then $\{A,B\}$ is a cell dumbbell that covers every pair covered by $\{X,Y\}$. In particular, $\thetacell(S)\le\theta_\eps(S)$ for every set $S$ of pairs of the band.
\end{lemma}
\begin{proof}
Let $\ell=\dist(X,Y)\ge\sigma/(1+2\eps)\ge\sigma/3$. Every $z\in\Gamma(A)$ lies in a cell $\Gamma(q)$ that contains a point $x\in X$, so $\dist(z,x)\le\dist(z,q)+\dist(q,x)\le2\eta$. Hence $\diam(\Gamma(A))\le\eps\ell+4\eta$, similarly for $B$, and $\dist(\Gamma(A),\Gamma(B))\ge\ell-4\eta$. Since $4\eta=\zeta\eps\sigma/12\le\zeta\eps\ell/4$, the dumbbell $\{\Gamma(A),\Gamma(B)\}$ is $\eps(1+\zeta/4)/(1-\zeta/4)$-separated, and $1+\zeta/4\le(1+2\zeta)(1-\zeta/4)$ for $\zeta\le1/2$. Finally, $X\subseteq\Gamma(A)$ and $Y\subseteq\Gamma(B)$.
\end{proof}

\begin{lemma}\label{lem:pack}
For every $z\in P$ and $R\ge\eta$, we have $|P_\eta\cap\ball(z,R)|\le(4R/\eta)^\kappa$.
\end{lemma}
\begin{proof}
Let $t=\lfloor\log_2(2R/\eta)\rfloor+1$. Applying the doubling property $t$ times covers $\ball(z,R)\cap P$ by $2^{\kappa t}\le(4R/\eta)^\kappa$ balls of radius $R/2^t<\eta/2$. Each of them contains at most one point of $P_\eta$.
\end{proof}

\begin{lemma}\label{lem:local-net}
Let $\rho\ge1$, let $u,w\in P$, and let $S$ be a set of pairs of the band whose endpoints lie in $\ball(u,R)\cup\ball(w,R)$, where $R=\rho\,\eps\Lambda\sigma$. Let $Q\subseteq P_\eta$ be the net points whose cells contain an endpoint of a pair of $S$. Then $|Q|\le K=2(384\rho\Lambda/\zeta)^\kappa$, and:
\begin{enumerate}
\item $\thetacell(S)\le K^2$;
\item $\thetacell(S)$ and an optimal family of cell dumbbells can be computed in $2^{O(K^3)}+O(4^Kn^2)$ time.
\end{enumerate}
\end{lemma}
\begin{proof}
Every $q\in Q$ is within distance $\eta\le R$ of an endpoint, so $Q\subseteq\ball(u,2R)\cup\ball(w,2R)$, and \cref{lem:pack} gives $|Q|\le2(8R/\eta)^\kappa=K$, since $R/\eta=48\rho\Lambda/\zeta$.

(1) Let $e\in S$ have endpoints in the cells of $q$ and $q'$. The set $\{\{q\},\{q'\}\}$ is a cell dumbbell covering $e$, because $\diam(\Gamma(q))\le2\eta=\zeta\eps\sigma/24$ and $\dist(\Gamma(q),\Gamma(q'))\ge\sigma-4\eta\ge\sigma/2$ (in particular, $q\ne q'$).

(2) A cell dumbbell that covers no pair of $S$ can be dropped. If a cell dumbbell $\{A,B\}$ covers a pair of $S$, then $A\cap Q$ and $B\cap Q$ are nonempty, and replacing $\{A,B\}$ by $\{A\cap Q,B\cap Q\}$ keeps it $\eps'$-separated by \cref{obs:hered} and keeps the pairs of $S$ that it covers, so we may assume $A,B\subseteq Q$. There are at most $4^K$ such candidates, and each is tested on the points in $O(n^2)$ time. Whether a candidate covers a pair of $S$ depends only on the cells of its endpoints, which gives at most $K^2$ types of pairs. By~(1), an optimal solution uses at most $K^2$ candidates, and exhaustive search over families of at most $K^2$ of the $4^K$ candidates takes $2^{O(K^3)}$ time.
\end{proof}

\begin{proof}[Proof of \cref{thm:doubling}]
All output dumbbells are cell dumbbells, hence $\eps'$-separated, and $\eps'=(1+2\delta/3)\eps$. The size bound follows from \cref{lem:bands,lem:cores}, with \cref{lem:round-net} in place of \cref{lem:round}, exactly as in the proof of \cref{thm:ptas}. For the running time, the proof of \cref{lem:radius}, with \cref{lem:local-net}(1) in place of \cref{lem:grid}, gives
\[
  (1+\zeta)^r<\thetacell(N^r(v))\le 4\big(384(r+1)\Lambda/\zeta\big)^{2\kappa},
\]
so $r\le r_{\max}$ for some $r_{\max}=O\big((\kappa/\zeta)(1/\zeta+\log(\kappa/\zeta))\big)$. Every set $N^{r}(v)$ that the algorithm evaluates has $r\le r_{\max}+1$ and satisfies the hypothesis of \cref{lem:local-net} with $\rho=r_{\max}+2$, and $\Lambda\le3^m$ gives $K=2^{O(\kappa(1/\zeta+\log(\kappa/\zeta)))}$. As in \cref{lem:grid}, the scale $\eps\sigma$ cancels in $R/\eta$, so with $\Lambda\le3^m$ the bound on $K$ does not depend on $\eps$. The remaining steps are as in the proof of \cref{thm:ptas}, and computing one net per nonempty band takes $O(n^4)$ time per offset, so the total time is $f(\kappa,\delta)\cdot n^6$. The algorithm never uses $\kappa$: the growth stops by itself, and the local solver enumerates the subsets of whichever set $Q$ it encounters.
\end{proof}

\section{Hardness in General Metrics}
\label{sec:general}

Our proof does not extend to general metrics, because the growth bound of \cref{lem:radius} requires bounded doubling dimension. In fact, relaxing the separation by a factor $\tau<2$ keeps the problem hard there. We show this with a direct reduction from clique cover, which works for every $\eps$ and $\tau\ge1$ with $\tau<2$ and $\tau\eps<1/2$, in particular for every $\tau<2$ when $\eps\le1/4$. The metric consists of copies of a graph: two vertices of a copy are at distance $1$ if they are adjacent and $2$ otherwise, and two vertices of different copies are at a large distance. A separated dumbbell between two copies then has cliques as sides, even if the separation is relaxed by a factor below $2$, so the minimum WSPD encodes clique covers of the graph. For $\tau=1$ and $\eps=1/3$, this is the setting of the hardness results of Buchin et al.~\cite[Thm.~28, Lemma~29]{BCHKLPRW26}, and our construction pads the instance with many copies of the input graph as in their Lemma~29. The new ingredients are the slack $\tau<2$, all $\eps<1/2$, and the explicit exponent $2/3$, which follows from the product bound~\eqref{eq:rbounds} and the hardness of the chromatic number.

For a graph $G$ with vertex set $V$, a \emph{clique} is a nonempty set of pairwise adjacent vertices (a single vertex is a clique). Let $\bar\chi(G)$ be the minimum number of cliques whose union is $V$, and let $\omega(G)$ be the maximum size of a clique. For two graphs $G_1$ and $G_2$ with vertex sets $V_1$ and $V_2$, let $\mathrm{cc}(G_1,G_2)$ be the minimum number of products $A\times B$ of a clique $A$ of $G_1$ and a clique $B$ of $G_2$ whose union is $V_1\times V_2$. If $G_2$ has a single vertex, then $\mathrm{cc}(G_1,G_2)=\bar\chi(G_1)$. Moreover,
\begin{equation}\label{eq:rbounds}
  \frac{|V_1|\,|V_2|}{\omega(G_1)\,\omega(G_2)}\;\le\; \mathrm{cc}(G_1,G_2)\;\le\;\bar\chi(G_1)\,\bar\chi(G_2),
\end{equation}
since each product covers at most $\omega(G_1)\,\omega(G_2)$ elements, and the products of two clique covers cover $V_1\times V_2$.

\subparagraph*{Block metrics.}
Let $G_1,\dots,G_k$ be graphs with pairwise disjoint vertex sets $V_1,\dots,V_k$, and let $\Delta\ge2$. The \emph{block metric} on $V_1\cup\dots\cup V_k$ is defined as follows. Two distinct vertices of the same block $V_i$ are at distance $1$ if they are adjacent in $G_i$ and at distance $2$ otherwise, and two vertices of different blocks are at distance $\Delta$ (\cref{fig:blocks}). This is a metric: if $x$ and $z$ lie in different blocks, then every $y$ lies in a block different from that of $x$ or from that of $z$, so $\dist(x,y)+\dist(y,z)\ge\Delta$; if $x$ and $z$ lie in the same block, then $\dist(x,z)\le2$, while $\dist(x,y)+\dist(y,z)\ge2$ for $y\notin\{x,z\}$.

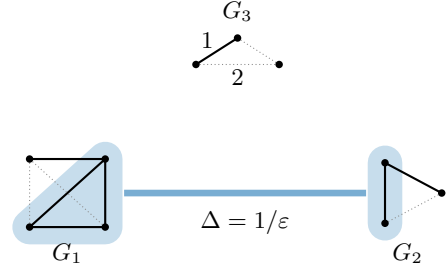
\begin{figure}[htb]
\centering
\begin{tikzpicture}[font=\small]
\draw[wspdblue!25, line width=13pt, line cap=round, line join=round] (-0.5,-0.45) -- (0.5,-0.45) -- (0.5,0.45) -- cycle;
\fill[wspdblue!25] (-0.5,-0.45) -- (0.5,-0.45) -- (0.5,0.45) -- cycle;
\draw[wspdblue!25, line width=13pt, line cap=round, line join=round] (4.2,-0.4) -- (4.2,0.4);
\draw[wspdblue!60, line width=2.5pt] (0.75,0) -- node[below=2pt, black] {$\Delta=1/\eps$} (3.95,0);
\draw[thick] (-0.5,-0.45) -- (0.5,-0.45);
\draw[thick] (-0.5,-0.45) -- (0.5,0.45);
\draw[black!45, densely dotted] (-0.5,-0.45) -- (-0.5,0.45);
\draw[thick] (0.5,-0.45) -- (0.5,0.45);
\draw[black!45, densely dotted] (0.5,-0.45) -- (-0.5,0.45);
\draw[thick] (0.5,0.45) -- (-0.5,0.45);
\fill[black] (-0.5,-0.45) circle[radius=1.4pt];
\fill[black] (0.5,-0.45) circle[radius=1.4pt];
\fill[black] (0.5,0.45) circle[radius=1.4pt];
\fill[black] (-0.5,0.45) circle[radius=1.4pt];
\node[below] at (0,-0.55) {$G_1$};
\draw[thick] (4.2,-0.4) -- (4.2,0.4);
\draw[black!45, densely dotted] (4.2,-0.4) -- (4.95,0);
\draw[thick] (4.2,0.4) -- (4.95,0);
\fill[black] (4.2,-0.4) circle[radius=1.4pt];
\fill[black] (4.2,0.4) circle[radius=1.4pt];
\fill[black] (4.95,0) circle[radius=1.4pt];
\node[below] at (4.5,-0.55) {$G_2$};
\draw[thick] (1.7,1.7) -- (2.25,2.05);
\draw[black!45, densely dotted] (1.7,1.7) -- (2.8,1.7);
\draw[black!45, densely dotted] (2.25,2.05) -- (2.8,1.7);
\fill[black] (1.7,1.7) circle[radius=1.4pt];
\fill[black] (2.25,2.05) circle[radius=1.4pt];
\fill[black] (2.8,1.7) circle[radius=1.4pt];
\node[above] at (2.25,2.15) {$G_3$};
\node[above left, inner sep=1pt] at (1.975,1.875) {$1$};
\node[below, inner sep=2pt] at (2.25,1.7) {$2$};
\end{tikzpicture}
\caption{A block metric and a dumbbell whose sides are cliques of $G_1$ and $G_2$.}
\label{fig:blocks}
\end{figure}

\begin{lemma}\label{lem:blocks}
Let $\eps>0$ and $\tau\in[1,2)$ with $\tau\eps<1/2$, and let $P$ be a block metric with parameter $\Delta=1/\eps$. For a dumbbell $\{X,Y\}$ of $P$, the following are equivalent: (i) $\{X,Y\}$ is $\tau\eps$-separated; (ii) $\{X,Y\}$ is $\eps$-separated; (iii) either $X$ and $Y$ are singletons contained in the same block, or there are $i\ne j$ such that $X$ is a clique of $G_i$ and $Y$ is a clique of $G_j$.
\end{lemma}
\begin{proof}
Note that $\Delta=1/\eps>2\tau\ge2$. (ii) implies (i) because $\tau\ge1$. Assume (i), and let $\ell=\dist(X,Y)\in\{1,2,\Delta\}$. If $\ell\le2$, then both sides have diameter at most $2\tau\eps<1$, so they are singletons, and since $\ell\le2$ they lie in the same block. If $\ell=\Delta$, then both sides have diameter at most $\tau\eps\Delta=\tau<2$. Two points at distance less than $2$ lie in the same block and are adjacent, so each side is a clique of one graph, and the two sides lie in different blocks because $\ell=\Delta$. Hence (iii) holds. Finally, assume (iii). Two singletons are $\eps$-separated. If $X$ and $Y$ are cliques of $G_i$ and $G_j$ with $i\ne j$, then all distances between $X$ and $Y$ are $\Delta$ and both diameters are at most $1=\eps\Delta$, so (ii) holds.
\end{proof}

\begin{lemma}\label{lem:blockopt}
Under the assumptions of \cref{lem:blocks}, every $\tau\eps$-WSPD of $P$ is an $\eps$-WSPD, and
\[
  \opt_\eps(P)\;=\;\sum_{i}\binom{|V_i|}2+\sum_{i<j}\mathrm{cc}(G_i,G_j).
\]
\end{lemma}
\begin{proof}
The first claim follows from \cref{lem:blocks}. By the same lemma, a pair of two points of $V_i$ is covered only by the dumbbell consisting of these two singletons, and a dumbbell that covers a pair with one point in $V_i$ and one in $V_j$, $i<j$, has a clique $A$ of $G_i$ as one side and a clique $B$ of $G_j$ as the other, and covers exactly the pairs corresponding to $A\times B$. Such dumbbells cover all pairs between $V_i$ and $V_j$ if and only if the corresponding products cover $V_i\times V_j$. Since every dumbbell covers pairs of only one of these parts, the minimum is the sum of the minima of the parts.
\end{proof}

Our main hardness result follows from the hardness of approximating the chromatic number. We use it in the following form: for every $\gamma\in(0,1/2)$, it is NP-hard, given a graph $G$ on $n$ vertices such that $\bar\chi(G)\le n^\gamma$ or $\omega(G)\le n^\gamma$, to decide which of the two cases holds. For the complement $\bar G$, the proof of Theorem~1.2 of Zuckerman~\cite{Zuc07}, which derandomizes a reduction of Feige and Kilian~\cite{FK98}, shows that it is NP-hard to distinguish $\chi_f(\bar G)\le n^{\gamma/2}$ from $\alpha(\bar G)\le n^{\gamma/2}$, where $\chi_f$ is the fractional chromatic number and $\alpha$ the independence number. Since $\chi(\bar G)\le(1+\ln n)\,\chi_f(\bar G)$ (Lov\'asz~\cite[Cor.~2 and Thm.~6]{Lov75}), and $\bar\chi(G)=\chi(\bar G)$ and $\omega(G)=\alpha(\bar G)$, the claim follows for $n$ large enough.

\begin{theorem}\label{thm:hard}
Let $\eps>0$ be rational, $\tau\in[1,2)$ with $\tau\eps<1/2$, and $c<2/3$. Unless P${}={}$NP, no polynomial-time algorithm computes, for every finite metric space $P$ on $N$ points, a $\tau\eps$-WSPD of $P$ of size at most $N^c\opt_\eps(P)$.
\end{theorem}
\begin{proof}
Since $c<2/3$, there is $\gamma\in(0,1/2)$ with $3c+4\gamma<2$. Given a graph $G$ on $n\ge2$ vertices, let $P$ be the block metric with parameter $\Delta=1/\eps$ and $\beta=n^2$ blocks, each a copy of $G$. It has $N=n^3$ points, so it can be built, and the algorithm run on it, in time polynomial in $n$. By \cref{lem:blockopt},
\[
  \opt_\eps(P)=\beta\binom n2+\binom\beta2\,\mathrm{cc}(G,G).
\]
If $\bar\chi(G)\le n^\gamma$, then $\mathrm{cc}(G,G)\le n^{2\gamma}$ by~\eqref{eq:rbounds}, and $\opt_\eps(P)\le n^4/2+n^{4+2\gamma}/2\le n^{4+2\gamma}$. If $\omega(G)\le n^\gamma$, then $\mathrm{cc}(G,G)\ge n^{2-2\gamma}$ by~\eqref{eq:rbounds}, and since $\binom\beta2\ge n^4/4$, every $\tau\eps$-WSPD of $P$, which is an $\eps$-WSPD by \cref{lem:blockopt}, has at least $n^{6-2\gamma}/4$ dumbbells. In the first case, the algorithm outputs at most $N^c n^{4+2\gamma}=n^{3c+4+2\gamma}$ dumbbells, which is less than $n^{6-2\gamma}/4$ as soon as $n^{2-3c-4\gamma}>4$. Hence, for $n$ above a constant, comparing the output size with $n^{6-2\gamma}/4$ decides which case holds, and smaller instances are solved by brute force.
\end{proof}

For $c=0$, \cref{thm:hard} states that no polynomial-time algorithm computes a $\tau\eps$-WSPD of size $\opt_\eps(P)$. This case needs only the NP-completeness of clique cover~\cite{Kar72}. Given a graph $G$ on $n$ vertices, let $P$ be the block metric with parameter $\Delta=1/\eps$ and two blocks, $G$ and a graph with a single vertex. By \cref{lem:blockopt}, a $\tau\eps$-WSPD of $P$ of size $\opt_\eps(P)$ has exactly $\binom n2+\bar\chi(G)$ dumbbells, so such an algorithm would compute $\bar\chi(G)$. For $\tau=1$ and $\eps=1/3$, this is the reduction of Buchin et al.~\cite[Thm.~28]{BCHKLPRW26}.

For $\tau<1$, a $\tau\eps$-WSPD is an $\eps$-WSPD, so for $\eps<1/2$ \cref{thm:hard} also holds for every $\tau<1$. The factor $2$ is inherent to block metrics: for $\tau\ge2$, a side at distance $\Delta=1/\eps$ may contain two nonadjacent vertices of a block, and the cliques no longer play a role. Metrics of this type cannot spread the two distances inside a block further apart either, because the triangle inequality forces the distance of two nonadjacent vertices with a common neighbor to be at most twice the edge length. The exponent $2/3$ is also the limit of this construction: with $n^b$ blocks instead of $n^2$, the same argument gives every $c<\min\{b,2\}/(b+1)$, which is largest for $b=2$.

In particular, for every $\delta\in(0,1]$ and every fixed rational $\eps\le1/4$, no polynomial-time algorithm achieves the separation factor $1+2\delta/3<2$ of \cref{thm:doubling} in general metrics with the size factor $N^c$ for any $c<2/3$, unless P${}={}$NP, let alone with the size factor $1+\delta$. Some polynomial size factor is always achievable: the $\binom N2$ dumbbells with singleton sides form a WSPD, and $\opt_\eps\ge N-1$ by \cref{lem:mst}, so the best exponent lies between $2/3$ and $1$. On the other hand, for $\eps<1/33$, Har-Peled et al.~\cite{HRR25} achieve the size factor $7$ with the separation factor $33$ in every metric space.

\section{Conclusion}
\label{sec:conclusion}

We gave a bicriteria approximation scheme for \minwspd in fixed dimension and in doubling metrics, and showed that in general metrics no such scheme exists, even for a fixed $\eps\le1/4$, unless P${}={}$NP.

Several questions remain open. Is there a PTAS that does not relax $\eps$, for instance via a separability lemma for optimal WSPDs in the plane? Can the approximation scheme be made to run in near-linear time? Does local search that exchanges up to $k$ dumbbells give a $(1+O(1/\sqrt k))$-approximation in the plane, as it does for covering problems with pseudodisks~\cite{MR10} and low-density objects~\cite{HQ17}, and for \minwspd on the line~\cite{BCHKLPRW26}? Is \minwspd NP-hard on the line, as asked by Buchin et al., or in the plane for a fixed $\eps$? Does an approximation scheme exist for the partition version? Finally, in general metrics and for small $\eps$, what is the smallest $\tau$ such that a $\tau\eps$-WSPD of size $O(\opt_\eps)$ can be computed in polynomial time? By \cref{thm:hard} and the result of Har-Peled et al.~\cite{HRR25}, it lies between $2$ and $33$. Going beyond $2$ requires metrics other than the block metrics of \cref{lem:blocks}.

\bibliography{refs}

\begin{thebibliography}{10}

\bibitem{AH12}
Mohammad~A. Abam and Sariel Har-Peled.
\newblock New constructions of {SSPDs} and their applications.
\newblock {\em Comput. Geom.}, 45(5--6):200--214, 2012.
\newblock \href {https://doi.org/10.1016/j.comgeo.2011.12.003}
  {\path{doi:10.1016/j.comgeo.2011.12.003}}.

\bibitem{Bar00}
Reuven Bar-Yehuda.
\newblock One for the price of two: A unified approach for approximating
  covering problems.
\newblock {\em Algorithmica}, 27(2):131--144, 2000.
\newblock \href {https://doi.org/10.1007/s004530010009}
  {\path{doi:10.1007/s004530010009}}.

\bibitem{BCHKLPRW26}
Kevin Buchin, Jacobus Conradi, Sariel Har-Peled, Antonia Kalb, Abhiruk Lahiri,
  Lukas Pl{\"a}tz, Carolin Rehs, and Sampson Wong.
\newblock On small pair decompositions for point sets.
\newblock In {\em Proc. 34th Annual European Symposium on Algorithms (ESA)},
  volume 388 of {\em LIPIcs}, pages 102:1--102:22, 2026.
\newblock Full version: arXiv:2601.22728.
\newblock \href {https://doi.org/10.4230/LIPIcs.ESA.2026.102}
  {\path{doi:10.4230/LIPIcs.ESA.2026.102}}.

\bibitem{CK93}
Paul~B. Callahan and S.~Rao Kosaraju.
\newblock Faster algorithms for some geometric graph problems in higher
  dimensions.
\newblock In {\em Proc. 4th Annual ACM-SIAM Symposium on Discrete Algorithms
  (SODA)}, pages 291--300, 1993.

\bibitem{CK95}
Paul~B. Callahan and S.~Rao Kosaraju.
\newblock A decomposition of multidimensional point sets with applications to
  $k$-nearest-neighbors and $n$-body potential fields.
\newblock {\em J. ACM}, 42(1):67--90, 1995.
\newblock \href {https://doi.org/10.1145/200836.200853}
  {\path{doi:10.1145/200836.200853}}.

\bibitem{CRW91}
Vasilis Capoyleas, G{\"u}nter Rote, and Gerhard~J. Woeginger.
\newblock Geometric clusterings.
\newblock {\em J. Algorithms}, 12(2):341--356, 1991.
\newblock \href {https://doi.org/10.1016/0196-6774(91)90007-L}
  {\path{doi:10.1016/0196-6774(91)90007-L}}.

\bibitem{CHHK14}
Parinya Chalermsook, Sandy Heydrich, Eugenia Holm, and Andreas Karrenbauer.
\newblock Nearly tight approximability results for minimum biclique cover and
  partition.
\newblock In {\em Proc. 22nd Annual European Symposium on Algorithms (ESA)},
  volume 8737 of {\em LNCS}, pages 235--246, 2014.
\newblock \href {https://doi.org/10.1007/978-3-662-44777-2_20}
  {\path{doi:10.1007/978-3-662-44777-2_20}}.

\bibitem{DGRSW25}
Lindsey Deryckere, Joachim Gudmundsson, Andr{\'e} van Renssen, Yuan Sha, and
  Sampson Wong.
\newblock A {WSPD}, separator and small tree cover for $c$-packed graphs.
\newblock In {\em Proc. 19th International Symposium on Algorithms and Data
  Structures (WADS)}, volume 349 of {\em LIPIcs}, pages 21:1--21:15, 2025.
\newblock \href {https://doi.org/10.4230/LIPIcs.WADS.2025.21}
  {\path{doi:10.4230/LIPIcs.WADS.2025.21}}.

\bibitem{DP11}
Adrian Dumitrescu and J{\'a}nos Pach.
\newblock Minimum clique partition in unit disk graphs.
\newblock {\em Graphs Combin.}, 27(3):399--411, 2011.
\newblock \href {https://doi.org/10.1007/s00373-011-1026-1}
  {\path{doi:10.1007/s00373-011-1026-1}}.

\bibitem{EJS05}
Thomas Erlebach, Klaus Jansen, and Eike Seidel.
\newblock Polynomial-time approximation schemes for geometric intersection
  graphs.
\newblock {\em SIAM J. Comput.}, 34(6):1302--1323, 2005.
\newblock \href {https://doi.org/10.1137/S0097539702402676}
  {\path{doi:10.1137/S0097539702402676}}.

\bibitem{ELS26}
Aryan Esmailpour, Khoi Le, and Stavros Sintos.
\newblock Compact representations of geometric bipartite graphs via weighted
  biclique covers, 2026.
\newblock arXiv:2608.20767.
\newblock \href {https://doi.org/10.48550/arXiv.2608.20767}
  {\path{doi:10.48550/arXiv.2608.20767}}.

\bibitem{FK98}
Uriel Feige and Joe Kilian.
\newblock Zero knowledge and the chromatic number.
\newblock {\em Journal of Computer and System Sciences}, 57(2):187--199, 1998.
\newblock \href {https://doi.org/10.1006/jcss.1998.1587}
  {\path{doi:10.1006/jcss.1998.1587}}.

\bibitem{GZ05}
Jie Gao and Li~Zhang.
\newblock Well-separated pair decomposition for the unit-disk graph metric and
  its applications.
\newblock {\em SIAM J. Comput.}, 35(1):151--169, 2005.
\newblock \href {https://doi.org/10.1137/S0097539703436357}
  {\path{doi:10.1137/S0097539703436357}}.

\bibitem{GP71}
Ronald~L. Graham and Henry~O. Pollak.
\newblock On the addressing problem for loop switching.
\newblock {\em Bell System Tech. J.}, 50(8):2495--2519, 1971.
\newblock \href {https://doi.org/10.1002/j.1538-7305.1971.tb02618.x}
  {\path{doi:10.1002/j.1538-7305.1971.tb02618.x}}.

\bibitem{GW26}
Joachim Gudmundsson and Sampson Wong.
\newblock A well-separated pair decomposition for low density graphs.
\newblock In {\em Proc. 2026 Annual ACM-SIAM Symposium on Discrete Algorithms
  (SODA)}, pages 3688--3723, 2026.
\newblock \href {https://doi.org/10.1137/1.9781611978971.135}
  {\path{doi:10.1137/1.9781611978971.135}}.

\bibitem{HL12}
Sariel Har-Peled and Mira Lee.
\newblock Weighted geometric set cover problems revisited.
\newblock {\em J. Comput. Geom.}, 3(1):65--85, 2012.
\newblock \href {https://doi.org/10.20382/jocg.v3i1a4}
  {\path{doi:10.20382/jocg.v3i1a4}}.

\bibitem{HM06}
Sariel Har-Peled and Manor Mendel.
\newblock Fast construction of nets in low-dimensional metrics and their
  applications.
\newblock {\em SIAM J. Comput.}, 35(5):1148--1184, 2006.
\newblock \href {https://doi.org/10.1137/S0097539704446281}
  {\path{doi:10.1137/S0097539704446281}}.

\bibitem{HQ17}
Sariel Har-Peled and Kent Quanrud.
\newblock Approximation algorithms for polynomial-expansion and low-density
  graphs.
\newblock {\em SIAM J. Comput.}, 46(6):1712--1744, 2017.
\newblock \href {https://doi.org/10.1137/16M1079336}
  {\path{doi:10.1137/16M1079336}}.

\bibitem{HRR25}
Sariel Har-Peled, Benjamin Raichel, and Eliot~W. Robson.
\newblock Well-separated pairs decomposition revisited, 2025.
\newblock arXiv:2509.05997.

\bibitem{HM85}
Dorit~S. Hochbaum and Wolfgang Maass.
\newblock Approximation schemes for covering and packing problems in image
  processing and {VLSI}.
\newblock {\em J. ACM}, 32(1):130--136, 1985.
\newblock \href {https://doi.org/10.1145/2455.214106}
  {\path{doi:10.1145/2455.214106}}.

\bibitem{Kar72}
Richard~M. Karp.
\newblock Reducibility among combinatorial problems.
\newblock In {\em Complexity of Computer Computations}, pages 85--103. Plenum
  Press, 1972.
\newblock \href {https://doi.org/10.1007/978-1-4684-2001-2_9}
  {\path{doi:10.1007/978-1-4684-2001-2_9}}.

\bibitem{Lov75}
L{\'a}szl{\'o} Lov{\'a}sz.
\newblock On the ratio of optimal integral and fractional covers.
\newblock {\em Discrete Mathematics}, 13(4):383--390, 1975.
\newblock \href {https://doi.org/10.1016/0012-365X(75)90058-8}
  {\path{doi:10.1016/0012-365X(75)90058-8}}.

\bibitem{MR10}
Nabil~H. Mustafa and Saurabh Ray.
\newblock Improved results on geometric hitting set problems.
\newblock {\em Discrete Comput. Geom.}, 44(4):883--895, 2010.
\newblock \href {https://doi.org/10.1007/s00454-010-9285-9}
  {\path{doi:10.1007/s00454-010-9285-9}}.

\bibitem{NS07}
Giri Narasimhan and Michiel H.~M. Smid.
\newblock {\em Geometric Spanner Networks}.
\newblock Cambridge University Press, 2007.
\newblock \href {https://doi.org/10.1017/CBO9780511546884}
  {\path{doi:10.1017/CBO9780511546884}}.

\bibitem{NH05}
Tim Nieberg and Johann Hurink.
\newblock A {PTAS} for the minimum dominating set problem in unit disk graphs.
\newblock In {\em Proc. 3rd Workshop on Approximation and Online Algorithms
  (WAOA)}, volume 3879 of {\em LNCS}, pages 296--306, 2006.
\newblock \href {https://doi.org/10.1007/11671411_23}
  {\path{doi:10.1007/11671411_23}}.

\bibitem{NHK08}
Tim Nieberg, Johann Hurink, and Walter Kern.
\newblock Approximation schemes for wireless networks.
\newblock {\em ACM Trans. Algorithms}, 4(4):49:1--49:17, 2008.
\newblock \href {https://doi.org/10.1145/1383369.1383380}
  {\path{doi:10.1145/1383369.1383380}}.

\bibitem{PS12}
Imran~A. Pirwani and Mohammad~R. Salavatipour.
\newblock A weakly robust {PTAS} for minimum clique partition in unit disk
  graphs.
\newblock {\em Algorithmica}, 62(3--4):1050--1072, 2012.
\newblock \href {https://doi.org/10.1007/s00453-011-9503-8}
  {\path{doi:10.1007/s00453-011-9503-8}}.

\bibitem{Sch79}
Arnold Sch{\"o}nhage.
\newblock On the power of random access machines.
\newblock In {\em Proc. 6th International Colloquium on Automata, Languages and
  Programming (ICALP)}, volume~71 of {\em LNCS}, pages 520--529. Springer,
  1979.
\newblock \href {https://doi.org/10.1007/3-540-09510-1_42}
  {\path{doi:10.1007/3-540-09510-1_42}}.

\bibitem{Smi18}
Michiel H.~M. Smid.
\newblock The well-separated pair decomposition and its applications.
\newblock In Teofilo~F. Gonzalez, editor, {\em Handbook of Approximation
  Algorithms and Metaheuristics, Volume 2: Contemporary and Emerging
  Applications}, chapter~4, pages 71--84. Chapman and Hall/CRC, 2nd edition,
  2018.
\newblock \href {https://doi.org/10.1201/9781351235426}
  {\path{doi:10.1201/9781351235426}}.

\bibitem{Var98}
Kasturi~R. Varadarajan.
\newblock A divide-and-conquer algorithm for min-cost perfect matching in the
  plane.
\newblock In {\em Proc. 39th Annual Symposium on Foundations of Computer
  Science (FOCS)}, pages 320--331, 1998.
\newblock \href {https://doi.org/10.1109/SFCS.1998.743466}
  {\path{doi:10.1109/SFCS.1998.743466}}.

\bibitem{Zuc07}
David Zuckerman.
\newblock Linear degree extractors and the inapproximability of max clique and
  chromatic number.
\newblock {\em Theory of Computing}, 3(6):103--128, 2007.
\newblock \href {https://doi.org/10.4086/toc.2007.v003a006}
  {\path{doi:10.4086/toc.2007.v003a006}}.

\end{thebibliography}
\end{document}